\documentclass[a4paper,11pt,runningheads,orivec]{llncs} 

\usepackage[english]{babel}
\usepackage[utf8]{inputenc}
\usepackage[T1]{fontenc}
\usepackage[pdftex,bookmarks=true,bookmarksnumbered=true,pdfpagemode=UseOutlines,pdfstartview=FitH,
    pdfpagelayout=OneColumn,colorlinks=false,urlcolor=blue,pdfborder={0 0 0}]{hyperref}
\usepackage{amsmath,amsfonts,amssymb}
\usepackage{complexity}

\usepackage{geometry}

\newcommand\RR{\mathbb{R}}
\newcommand\SPS{\mathsf{SPS}}
\newcommand\mSPS{\mathsf{mSPS}}

\DeclareMathOperator{\PER}{PER}
\DeclareMathOperator{\PW}{PW}

\begin{document}

\title{The Limited Power of Powering:\\ 
Polynomial Identity Testing and 
a 
Depth-four Lower Bound for the Permanent}
\author{Bruno Grenet\inst{1}\inst{3}, Pascal Koiran\inst{1}\inst{3}, Natacha Portier\inst{1}\inst{3}\thanks{This material is based on work supported in part by the European Community under contract PIOF-GA-2009-236197 of the 7th PCRD.}, Yann Strozecki\inst{2}\inst{3}}

\institute{
    LIP, UMR 5668, \'ENS de Lyon -- CNRS -- UCBL -- INRIA\\
    \'Ecole Normale Sup\'erieure de Lyon, Universit\'e de Lyon\\
    \url{[Bruno.Grenet,Pascal.Koiran,Natacha.Portier]@ens-lyon.fr}
\and
    \'Equipe de Logique Math\'ematique, Universit\'e Paris VII\\
    \url{Strozecki@logique.jussieu.fr}
\and
    Department of Computer Science, University of Toronto
}

\titlerunning{The Limited Power of Powering}
\authorrunning{B. Grenet, P. Koiran, N. Portier, Y. Strozecki}

\maketitle 

\begin{abstract}
Polynomial identity testing and arithmetic 
circuit 
lower 
bounds are two central
questions in algebraic complexity theory. It is an intriguing fact that these 
questions are actually related. 
One of the authors of the present paper has recently proposed 
a ``real $\tau$-conjecture'' which is inspired by this connection.
The real $\tau$-conjecture states that the number of real roots of 
a sum of products of sparse univariate polynomials should be 
polynomially bounded. It implies a superpolynomial lower bound on the
size of arithmetic circuits computing the permanent polynomial.

In this paper we show that the real-$\tau$ conjecture holds
 true for a restricted
class of sums of products of sparse polynomials.
This result yields lower bounds for a restricted class of depth-4 circuits:
we show that polynomial size circuits from this class cannot compute the
permanent, and we also give a deterministic polynomial identity testing
algorithm for the same class of circuits.

\end{abstract}

\section{Introduction}

The $\tau$-conjecture~\cite{ShubSmale,smale1998mathematical} 
states that a 
univariate
polynomial with integer coefficients defined 
by an arithmetic circuit 
has a number of integer roots polynomial in the size of the circuit.
A real version of this conjecture was recently 
presented in~\cite{koiran2011shallow}.
The real $\tau$-conjecture states that the number of real roots of 
a sum of products of sparse univariate polynomials should be 
polynomially bounded as a function of the size of the corresponding 
expression.
More precisely, consider a polynomial of the form
$$f(X)=\sum_{i=1}^k \prod_{j=1}^m f_{ij}(X),$$
where $f_{ij} \in \RR[X]$ has at most $t$ monomials.
The conjecture asserts that the number of real roots of $f$ is
bounded by a polynomial function of $kmt$.
It was shown in~\cite{koiran2011shallow} that this conjecture implies 
a superpolynomial lower bound on the arithmetic circuit complexity 
of the permanent polynomial (a central goal of algebraic complexity theory
ever since Valiant's seminal work~\cite{valiant1979completeness}).
In this paper we show that the conjecture holds true in a special case.
We focus on the case where the number of distinct sparse polynomials is small
(but each polynomial may  be repeated many times).
We therefore consider expressions of the form
\begin{equation} \label{expression}
\sum_{i=1}^k \prod_{j=1}^m f_{j}^{\alpha_{ij}}(X).
\end{equation}
We obtain a $O(t^{m(2^{k-1}-1)})$ upper bound on 
the number of real roots of such a polynomial, 
where $t$ is the maximum number of monomials
in the~$f_j$. In particular, the bound is polynomial in $t$ when the ``top fan-in'' $k$ and the number $m$ of sparse polynomials in the expression are
both constant. Note also that the bound is independent of the magnitude of
the integers~$\alpha_{ij}$.

From this upper bound we obtain a lower bound on the complexity 
of the permanent for a restricted class of arithmetic circuits.
The circuits that we consider are again of form~(\ref{expression}), but
now $X$ should be interpreted as the tuple of inputs to the circuit
rather than as a single real variable. Roughly speaking, we show 
a superpolynomial lower bound on the complexity of the permanent 
in the case where $k$ and $m$ are again fixed.
Note that this is a lower bound for a restricted class of depth-$4$ circuits:
the output gate at depth $4$ has fan-in bounded by the constant $k$, 
and the gates at depth $2$ are only allowed to compute a constant
($m$) number of distinct 
polynomials $f_j$.

Our third main result is a deterministic identity testing algorithm, again for
polynomials of the same form. When $k$ and $m$ are fixed, 
we can test if the polynomial in~(\ref{expression}) is identically equal to 0 
in time polynomial in $t$ and in $\max_{ij} \alpha_{ij}$.
Note that if $k$, $m$ {\em and the exponents~$\alpha_{ij}$} are all bounded 
by a constant then the number of monomials in such a polynomial
is $t^{O(1)}$ and our three main results become trivial.
These results are therefore interesting only in the case where the $\alpha_{ij}$ may be large, and can be interpreted as limits on the power of powering.

\subsection{Connection to Previous Work} \label{connections}

The idea of deriving lower bounds on arithmetic circuit complexity
from upper bounds on the number of real roots goes back at least
to a 1976 paper by Borodin and Cook~\cite{BC76}. Their results 
were independently improved by Grigoriev and Risler 
(see~\cite{BCS}, chapter~12). For a long time, it seemed that the lower bounds
that can be obtained by this method had to be rather small since the number
of real roots of a polynomial can be exponential in its arithmetic circuit
size. Nevertheless, as explained 
above it was recently shown in~\cite{koiran2011shallow} that superpolynomial
lower bounds on the complexity of the permanent on general arithmetic circuits
can be derived from a suitable upper bound on the number of roots of
sums of products of sparse polynomials.
This is related to the fact that for low degree polynomials,
arithmetic circuits of depth 4 are almost equivalent to general
arithmetic circuits \cite{agrawal2008arithmetic,koiran2010arithmetic}.

The study of polynomial identity testing (PIT) also has a long history.
The Schwartz-Zippel lemma~\cite{Schw}
 yields a randomized algorithm for PIT. 
 
A connection between deterministic PIT and arithmetic circuit lower bounds 
was pointed out as early as 1980 by Heintz and Schnorr~\cite{HS82}, 
but a more in-depth study of this connection began 
only much later~\cite{kabanets2004derandomizing}.
The recent literature contains deterministic PIT algorithms for various
restricted models (see e.g. the two surveys~\cite{AS09,Sax09}).
One model which is similar to ours was recently studied 
in~\cite{beecken2011algebraic}.
It follows from Theorem~1 in~\cite{beecken2011algebraic} that there is 
a polynomial time deterministic black-box PIT algorithm for polynomials 
of the form~(\ref{expression}) if, 
instead of bounding $k$ and $m$ as in our algorithm, we bound
the transcendence degree $r$ of the polynomials $f_j$. 
Obviously we have $r \leq m$, so from this point of view their result is more 
general.\footnote{As pointed out by the authors of~\cite{beecken2011algebraic}, their result already seems nontrivial for a constant~$m$.}
On the other hand their running time is polynomial in the degree
of the $f_j$, whereas we can handle polynomials of exponential degree
in polynomial time. Note also that~\cite{beecken2011algebraic} does not 
provide any lower bound result.

\subsection{Our approach} \label{approach}

The proof of 
our bound on the number of real roots has the same high-level structure as that
of Descartes' rule of signs. 
\begin{proposition} \label{descartes}
A univariate polynomial $f \in \RR[X]$ 
with $t \geq 1$ monomials has at most $t-1$ positive
real roots.  
\end{proposition}
The number of negative roots of $f$ is also bounded by $t-1$ 
(consider $f(-X)$),
hence there are at most $2t-1$ real roots (including 0).
There is also a refined version of Proposition~\ref{descartes} where 
the number of monomials $t$ is replaced by the number of sign changes in 
the sequence of coefficients of $f$. The cruder version will be sufficient
for our purposes.

We briefly recall an inductive proof of Proposition~\ref{descartes}.
For $t=1$, there is no non-zero root.
For $t>1$, let $a_{\alpha}X^{\alpha}$ be the monomial of lowest degree.
We can assume that $\alpha=0$ (if not, we can divide $f$ by $X^{\alpha}$ since this operation does not change the number of positive roots). 
Consider now the derivative~$f'$.
It has $t-1$ monomials, and at most $t-2$ positive real roots by induction
hypothesis.
Moreover, by Rolle's theorem there is a positive root of $f'$  
between 2 consecutive positive roots of $f$. We conclude that $f$ has
at most $(t-2)+1=t-1$ positive roots.

In~(\ref{expression}) we have a sum of $k$ terms instead of $t$ monomials, but
the basic strategy remains the same: we divide by the first term and
take the derivative. This has the effect of removing a term, but it also
has the effect (unlike Descartes' rule) of increasing the complexity of
the remaining $k-1$ terms. This results in a larger bound (and a longer proof).

From this upper bound we obtain our permanent lower bound by applying the proof
method which was put forward in~\cite{koiran2011shallow}.
More precisely, assume that the permanent has an
efficient representation of the form~(\ref{expression}).
We show that the same must be true for the univariate polynomial 
$\prod_{i=1}^{2^n} (X-i)$ using a result of 
B\"urgisser~\cite{burgisser2009defining}.
This yields a contradiction with our upper bound on the number of real roots.

Our third result is a polynomial identity testing algorithm.
Using a standard substitution technique, we can assume that the polynomials
$f_j$ in~(\ref{expression}) are univariate.
We note that the resulting $f_j$ may be of exponential degree even if 
the original multivariate $f_j$ are of low degree.
The construction of hitting sets is is a classical approach to deterministic
identity testing.
Recall that a hitting set for a class $\cal F$ of polynomials
is a set of points $H$ such that for any non-identically zero polynomial
$f \in \cal F$ we have a point $x \in H$ such that $f(x) \neq 0$.
Clearly, a hitting set yields a black-box identity testing algorithm
(it is not hard to see that the converse is also true).
Moreover, for any class $\cal F$ of univariate polynomials, 
an upper bound $z(\cal F)$ on the number of real roots of each non-zero 
polynomial in $\cal F$ yields a hitting set (any set of $z({\cal F})+1$ 
real numbers will do).
From our upper bound result we therefore have polynomial size hitting sets
for polynomials of the form~(\ref{expression}) when $k$ and $m$ are fixed.
Unfortunately, the resulting black-box algorithm does not run in polynomial
time: evaluating a polynomial at a point of the hitting set may not be
feasible in polynomial time since (as explained above) the $f_j$ may be
of very high degree. 
We therefore use a different strategy. Roughly speaking, we ``run'' the
proof of our upper bound theorem on an input of form~(\ref{expression}).
This requires explicit knowledge of this representation, and the resulting
algorithm is non-black-box.
As explained in Section~\ref{connections},  
for the case where the $f_j$ are low-degree multivariate polynomials
an efficient black-box algorithm was recently given in~\cite{beecken2011algebraic}.

\paragraph{Organization of the paper.}

In  Section~\ref{sec:bound} we prove an upper bound on the number of real
roots of polynomials of the form~\eqref{expression}, see Theorems~\ref{thm:naive} and~\ref{thm:support}. 
In fact, we obtain an upper bound for a more general class of polynomials
which we call $\SPS(k,m,t,h)$. This generalization is needed for the
inductive proof to go through.
From this upper bound, we derive in Section \ref{sec:lower}
a lower bound on the computational power of (multivariate) circuits 
of the same form.
We give  in Section \ref{sec:PIT} 
a deterministic identity testing algorithm,
again for polynomials of form~\eqref{expression}.

\section{The real roots of 
a
sum of 
products of sparse polynomials} \label{sec:bound}

\subsection{Definitions}\label{sec:def}

In this section, we define precisely the polynomials we are working with.
We then explain how to transform those polynomials in a way which 
reduces the number of terms 
but does not increase too much the number of roots.
This method has some similarities with the proof of Lemma $2$ in~\cite{li2003counting} and it leads
to a bound on the number of roots of the polynomials we study.

We say that a polynomial is \emph{$t$-sparse} if it has at most $t$ monomials.

\begin{definition} \label{SPS}
Let $\SPS(k,m,t,h)$ denote the class of polynomials $\phi\in\RR[X]$ defined by
\[\phi(X) = \sum_{i=1}^k g_i(X)\prod_{j=1}^m f_j^{\alpha_{ij}}(X)\]
where
\begin{itemize}
\item $g_1,\dots,g_k$ are $h$-sparse polynomials over $\RR$;
\item $f_1,\dots,f_m$ are $t$-sparse non-zero polynomials over $\RR$;
\item $\alpha_{11},\dots,\alpha_{km}$ are non-negative integers.
\end{itemize}
We define $P_i=\prod_{j=1}^m f_j^{\alpha_{ij}}$ and $T_i=g_iP_i$ for all $i$.
We also define $\pi = \prod_{j=1}^m f_j$.
Finally, we define $\SPS(k,m,t)$ as the subclass of $\SPS(k,m,t,h)$ 
in which all the $g_i$ are equal to the constant 1. 
\end{definition}
Note that $\SPS(k,m,t)$ is just the class of polynomials of 
form~\eqref{expression}, and is included in $\SPS(k,m,t,1)$.
We want to give a bound for the number of real roots of the polynomials in 
this class,
and more generally in $\SPS(k,m,t,h)$.
To this end, from a polynomial $\phi\in\SPS(k,m,t,h)$, we build a new polynomial $\tilde\phi\in\SPS(k-1,m,t,\tilde h)$ 
for some $\tilde h$ such that a bound on the number of real roots of $\tilde\phi$ yields a bound for $\phi$.

\begin{lemma}\label{tildeIsSps}
Let $\phi\in\SPS(k,m,t,h)$. If $g_1$ is not identically zero, we write $\tilde{\phi}=g_1T_1\pi(\phi/T_1)'$
otherwise $\tilde{\phi}=\phi$.
There exists $\tilde h$ such that
$\tilde{\phi}\in\SPS(k-1,m,t,\tilde h).$
\end{lemma}

\begin{proof}
If $g_1$ is identically zero, the theorem holds with $\tilde{h} = h$.
Assume now that $g_1$ is not identically zero and let 
\begin{equation*}\psi(X) = \phi(X)/T_1(X) = 1 + \frac{1}{T_1(X)}\cdot\sum_{i=2}^k T_i(X).\end{equation*}

Then 
\[\psi'=\frac{\sum_{i=2}^k \left( T_1T_i' - T_1'T_i\right)}{T_1^2}.\]
Notice that $T_i'=g_i'P_i + g_iP_i'$ and
\[P_i'=\sum_{j=1}^m \alpha_{ij}f_j' f_j^{\alpha_{ij}-1}\cdot\prod_{l\neq j} f_{l}^{\alpha_{il}}=P_i\cdot\sum_{j=1}^m \alpha_{ij} f'_j/f_j.\]
Therefore
\begin{align*}
\psi'&=\frac{1}{T_1^2}\cdot\sum_{i=2}^k (g_1P_1g_i'P_i+g_1P_1g_iP_i'-g_1'P_1g_iP_i-g_1P_1'g_iP_i)\\
&= \frac{1}{T_1^2}\cdot\sum_{i=2}^k(g_1g_i'P_1P_i+g_1g_iP_1P_i\sum_j \alpha_{ij} f_j'/f_j \\
& \qquad\qquad\qquad - g_1'g_iP_1P_i-g_1g_iP_1P_i\sum_j\alpha_{1j}f_j'/f_j)\\
&= \frac{1}{g_1T_1}\cdot\sum_{i=2}^k P_i\left(g_1g_i'-g_1'g_i + g_1g_i\sum_j(\alpha_{ij}-\alpha_{1j})f_j'/f_j\right).
\end{align*}

We now multiply $\psi'$ by $\pi=\prod_j f_j$ and get 
\begin{equation*}
\pi\psi' = \frac{1}{g_1T_1}\cdot\sum_{i=2}^k P_i\left(\pi\cdot(g_1g_i'-g_1'g_i)+g_1g_i\sum_j(\alpha_{ij}-\alpha_{1j})f_j'\prod_{l\neq j} f_l\right).
\end{equation*}

Thus $g_1T_1\pi\psi'$ is a polynomial of the class $\SPS(k-1,m,t,\tilde h)$ for some $\tilde h$. Let us write 
\[\tilde\phi = g_1T_1\pi\psi' = \sum_{i=2}^k P_i\tilde g_i.\]
The integer $\tilde h$ denotes the maximum number of monomials in $\tilde g_i$ for $2\le i\le k$. 
\qed\end{proof} 

\begin{definition}\label{def:phin}
Let $(\phi_n)_{1\le n\le k}$ be the sequence defined by $\phi_1=\phi$ and for $n\ge 1$, $\phi_{n+1}=\tilde\phi_n$. Let also, for $1\le i\le k$, $(g_i^{(n)})_{1\le n \le i}$ be defined by $g^{(1)}_i=g_i$ and $g^{(n+1)}_i=\widetilde{g^{(n)}_i}$ for $i>n$. In other words
\[\phi_n=\sum_{i=n}^k g^{(n)}_i\prod_{j=1}^m f_j^{\alpha_{ij}}.\]
We also define the sequence $(h_n)_{1\le n\le k}$ by $h_1=1$ and $h_{n+1}=\tilde h_n$. That is, each $g^{(n)}_i$ is $h_n$-sparse.
\end{definition}

\subsection{A generalization of Descartes' rule} \label{sec:first}

In Definition~\ref{def:phin} we defined 
a sequence of polynomials $(\phi_n)$ and  a sequence of integers $(h_n)$.
In this section we first prove that the number of real roots of $\phi_n$ 
is bounded by the number of real roots of $\phi_{n+1}$ up to a multiplicative
 constant.
Then, we give an upper bound on $h_n$
and we combine these ingredients to obtain a bound on the number of
real roots of a polynomial in $\SPS(k,m,t)$.
This bound (in Theorem~\ref{thm:naive} at the end of the section)
is polynomial in $t$.

We denote by $r(P)$ the number of distinct real roots 
of a rational function~$P$.
In order to obtain a bound on $r(\phi)$ from a bound on $r(\tilde\phi)$,
we need the following lemma.

\begin{lemma}\label{BoundSPS1}
Let $P\in\SPS(1,m,t,h)$. If $P$ is not identically zero then 
\[r(P)\le 2h+2m(t-1)-1.\]
\end{lemma}

\begin{proof}
By definition, $P=g\cdot\prod_j f_j^{\alpha_{j}}$. 
The number of non-zero real roots of $P$ is therefore bounded by the sum
of the number of non-zero real roots of $g$ and of the $f_j$'s.
Since $g$ is $h$ sparse, we know from Descartes' rule that is has at most
$2(h-1)$ non-zero real roots. Likewise, each $f_j$ has at most
$2(t-1)$ real roots.
As a result, $P$ has at most $2(h-1)+2m(t-1)$ non-zero real roots.
Since $0$ can also be a root, we add $1$ to this bound to obtain the final result.
\qed\end{proof}

\begin{lemma}\label{Boundphi}
Let $\phi\in\SPS(k,m,t,h)$. Then 
\[ r(\phi) \le r(\tilde\phi)+4h+4m(t-1)-1.\]
\end{lemma}

\begin{proof}
If $g_1$ is zero in the definition of $\phi$, then $\tilde{\phi}=\phi$ which proves the
lemma.
 
Recall from the proof of Lemma~\ref{tildeIsSps} the notation $\psi=\phi/T_1$. 
If $g_1$ is not identically zero, by definition we have $\tilde{\phi}=g_1T_1\pi\psi'$, 
so the  number $r(\tilde\phi)$ of real roots of the polynomial $\tilde\phi$ is an upper bound on the number of real roots of $\psi'$. 

Since $\phi=T_1 \psi$, we have $r(\phi) \leq r(T_1)+r(\psi)$. 
Moreover, between two consecutive roots of the rational function 
$\psi$, we have a root of $\psi'$ or a root of the denominator $T_1$.
As a result, $r(\psi) \leq r(\psi')+r(T_1)+1$.
It follows that $r(\phi)\le r(\psi')+2r(T_1)+1\le r(\tilde\phi)+2r(T_1)+1$.
Moreover, the polynomial $T_1=g_1\cdot\prod_j f_j^{\alpha_{1j}}$ is in $\SPS(1,m,t,h)$. 
Thus by Lemma~\ref{BoundSPS1}, 
$T_1$ has at most $2h+2m(t-1)-1$ real roots.
We conclude that $\phi$ has at most
\begin{equation*}
r(\tilde\phi) + 2\cdot\left( 2h + 2m(t-1) - 1 \right) + 1 = r(\tilde\phi) + 4h + 4m(t-1) - 1
\end{equation*}
real roots. 
\qed\end{proof} 

\begin{proposition}\label{prop:Boundrphi}
Let $\phi\in\SPS(k,m,t,1)$. Then 
\[r(\phi)\le 2h_k+4\sum_{i=1}^{k-1} h_i + 2m(2k-1)(t-1)-k.\]
\end{proposition}

\begin{proof}
Lemma~\ref{Boundphi} gives the following recurrence:
\begin{equation*}
r(\phi_n)\le r(\phi_{n+1})+4h_n+4m(t-1)-1.
\end{equation*}
Thus, we get
\begin{equation} \label{eq:rec}
r(\phi) \le r(\phi_k)+4\sum_{i=1}^{k-1} h_i + (k-1)(4m(t-1)-1).
\end{equation}
Since $\phi_k\in\SPS(1,m,t,h_k)$, Lemma~\ref{BoundSPS1} bounds its number of real roots:  
\begin{equation} \label{eq:phik}
r(\phi_k)\le 2h_k+2m(t-1)-1.
\end{equation}
The bound is a combination of \eqref{eq:rec} and \eqref{eq:phik}.
\qed\end{proof} 

Proposition~\ref{prop:Boundrphi} shows that in order to bound $r(\phi)$, we need a bound on $h_n$.

\begin{proposition}\label{prop:Boundh}
For all $n$, 
$h_n$ is bounded by $((m+2)t^m)^{2^{n-1}-1}$.
\end{proposition}

\begin{proof}
As showed in the proof of Lemma~\ref{tildeIsSps}, $\tilde\phi = \sum_{i=2}^k \tilde g_i P_i$
where each $\tilde g_i$ is $\tilde h$-sparse. More precisely,  
\[\tilde g_i = (g_1g_i'-g_1'g_i)\prod_{j=1}^m f_j+g_1g_i\sum_{j=1}^m(\alpha_{ij}-\alpha_{1j})f_j'\prod_{l\neq j} f_l.\]
Thus $\tilde g_i$ is a sum of $(m+2)$ terms, and each term is a product of $m$ $t$-sparse polynomials by two $h$-sparse polynomials. Thus $\tilde h \le (m+2)t^mh^2$. 

This gives the following recurrence relation on $h_n$:
\[\begin{cases}
    h_1 &=1\\
    h_{n+1} &\le (m+2)t^mh_n^2
\end{cases}\]
Therefore, $h_n\le ((m+2)t^m)^{2^{n-1}-1}$.
\qed\end{proof} 

Now, we combine Propositions~\ref{prop:Boundrphi} and \ref{prop:Boundh}
to obtain our first bound on the number of roots of a polynomial in 
$\SPS(k,m,t)$.

\begin{theorem}\label{thm:naive}
Let $\phi\in 
\SPS(k,m,t)$: we have $\phi=\sum_{i=1}^k\prod_{j=1}^m f_j^{\alpha_{ij}}$ where for all $i$ and $j$, $f_j$ is $t$-sparse and $\alpha_{ij}\ge 0$. Then
$r(\phi)\le  C\times((m+2)t^m)^{2^{k-1}-1}$ 
for some 
universal
constant $C$. 
\end{theorem}

\begin{proof}
It follows from Propositions~\ref{prop:Boundrphi} and~\ref{prop:Boundh}
that the number of real roots of a polynomial $\phi\in\SPS(k,m,t,1)$ is
\[r(\phi)\le 2((m+2)t^m)^{2^{k-1}-1}+4\sum_{i=1}^{k-1}((m+2)t^m)^{2^{i-1}-1}  + 2m(2k-1)(t-1)-k.\] 
To simplify this expression, note that
\[\sum_{i=1}^{k-1}((m+2)t^m)^{2^{i-1}-1}\le (k-1)((m+2)t^m)^{2^{k-2}-1}.\]
It is then clear that the function $((m+2)t^m) ^{2^{k-1}-1}$ dominates the two smallest terms in the 
bound on $r(\phi)$.
The result follows since $\SPS(k,m,t) \subseteq \SPS(k,m,t,1)$.
\qed\end{proof}

\subsection{A tighter analysis}\label{sec:improvment}

This 
section is devoted to an improved bound for $h_n$, the number of monomials in the polynomials $g_i^{(n)}$.
That automatically sharpens the bound we give for the number of real roots of a polynomial in $\SPS(k,m,t)$. 

Let $P$ be a polynomial, and let $S(P)$ be
its support, that is the set of integers $i$ such that
$X^i$ has a nonzero coefficient 
in $P$.
Let $A$ be a set of integers, we write $A-\bf{1}$ for the set  $\{ i-1 \mid i \in A\}$.
If $A$ and $B$ are two sets, we write $A+B$ for the set $\{i+j \mid i\in A, j \in B\}$
and we write $n \times A$ for the sum of $n$ copies of the set 
$A$.
Remark that the sum is commutative 
and 
that $A+(B-\textbf{1}) = (A-\textbf{1})+B$. 
We shall use some easy properties of the supports of polynomials. The proof is left to the reader.
\begin{lemma}\label{observation}
Let $P$ and $Q$ be two polynomials, then 
\begin{enumerate}
\item $S(P') \subseteq  S(P)-\bf{1}$;
\item $S(P+Q) \subseteq S(P) \cup S(Q)$;
\item $S(PQ) \subseteq S(P)+S(Q)$.
\end{enumerate}
\end{lemma}

Now consider a 
polynomial $\phi\in\SPS(k,m,t)$ as in the previous section. 
Recall that $\phi_n =  \sum_{i=n}^k g_i^{(n)} P_i$ is
the polynomial obtained from $\phi$ after $n$ steps of the transformation in the first section.
Let $ S$ be the set $(\sum_j S(f_j)) -\textbf{1}$. We prove by induction on $n$ that for all $i>n$, $g_i^{(n)}$ satisfies $S(g_i^{(n)}) \subseteq (2^n-1) \times S$. To this end, we prove the following lemma.

\begin{lemma}\label{lemma:support}
Let $\phi\in\SPS(k,m,t,h)$, and $\tilde\phi\in\SPS(k-1,m,t,\tilde h)$ as defined in Lemma~\ref{tildeIsSps}. Then
\[\bigcup_{i=2}^k S(\tilde g_i)\subseteq 2\times \left(\bigcup_{i=1}^k S(g_i)\right)+S.\]
\end{lemma}

\begin{proof}
To simplify notations, let us define $S_g=\bigcup_i S(g_i)$ and $S_{\tilde g}=\bigcup_i S(\tilde g_i)$. We aim to show that $S_{\tilde g}\subseteq 2\times S_g +S$.

Recall that \[\tilde{g}_i = \pi\cdot(g_ng_i'-g_n'g_i)+g_ng_i\sum_j(\alpha_{ij}-\alpha_{nj})f_j'\prod_{l\neq j} f_l.\]
Applying Lemma~\ref{observation}(2) yields 
\[S(\tilde{g}_i) \subseteq S(\pi g_ng_i') \cup S(\pi g_n'g_i) \cup S\biggl(g_ng_i\sum_j(\alpha_{ij}-\alpha_{nj})f_j'\prod_{l\neq j} f_l\biggr).\]
By Lemma~\ref{observation}(3), we have 
\[S(\pi g_ng_i') \subseteq  S(\pi)+S(g_n)+S(g_i').\]
Moreover, $S(g_n)\subseteq S_g$ and $S(g_i')\subseteq \bigcup_i (S(g_i)-\mathbf 1)=S_g-\mathbf 1$. Thus
\[S(\pi g_ng_i')\subseteq S(\pi)+ S_g+(S_g-\mathbf 1).\]
Since $-\textbf{1}$ commutes with $+$, we obtain:
\[S(\pi g_ng_i') \subseteq (S(\pi)-\textbf{1}) + 2 \times S_g.\]
Now, $S(\pi)-\textbf{1}=S$ by definition, and $S(\pi g_ng_i') \subseteq S + 2 \times S_g$.
The proof is the same for $S(\pi g_n'g_i)\subseteq S+ 2 \times S_g$.

Finally, it holds that \[ S\biggl(g_ng_i\sum_j(\alpha_{ij}-\alpha_{nj})f_j'\prod_{l\neq j} f_l\biggr) \subseteq 2 \times S_g + \bigcup_j S(f_j'\prod_{l\neq j} f_l).\]
Furthermore,
\[\bigcup_j S(f_j'\prod_{l\neq j} f_l) \subseteq  \bigcup_j \biggl((S(f_j) - \textbf{1})+\sum_{l\neq j} S(f_l)\biggr) = S.\]
Therefore we have \[  S\biggl(g_ng_i\sum_j(\alpha_{ij}-\alpha_{nj})f_j'\prod_{l\neq j} f_l\biggr) \subseteq S+2 \times S_g. \]

We proved that for every $i>n$, $S(\tilde g_i)\subseteq S+2\times S_g$. This is enough to conclude that 
\[S_{\tilde g}\subseteq S+ 2\times S_g.\]
\qed\end{proof} 

\begin{proposition}\label{prop:Sgi}
Let $\phi\in\SPS(k,m,t)$ and let $\phi_n$ be defined as in Definition~\ref{def:phin}. Then for $1\le n\le i\le k$, 
\[S(g_i^{(n)})\subseteq (2^{n-1}-1)\times S.\]
\end{proposition}

\begin{proof}
We actually show by induction on $n$ that $\bigcup_{i\ge n} S(g_i^{(n)})\subseteq (2^{n-1}-1)\times S$. For $n=1$, it is clear since the $g_i^{(1)}$ have degree $0$.
By definition $g_i^{(n+1)}=\widetilde{ g_i^{(n)} }$, thus
 Lemma~\ref{lemma:support} proves the induction step.
\qed\end{proof}

We need the following combinatorial lemma to improve the bound of Theorem~\ref{thm:naive}.

\begin{lemma}\label{lemma:S}
Let $S$ be a set of integers and $p>0$. Then 
\[\left|p\times S\right|\le\binom{p+|S|}{p}\le \left[e\times\left(1+\frac{|S|}{p}\right)\right]^p.\]
\end{lemma}

\begin{proof}
We want to count the number of different sums of $p$ terms from $S$. This is bounded from above by the number of non-decreasing sequences of elements from $S$ of length $p$ (where elements can be repeated). 
To count such non-decreasing sequences, we can assume without loss of generality that $S=\{1,\dots,N\}$
where $N=|S|$. To a non-decreasing sequence $(s_1,\dots, s_p)$, we associate the sequence $(t_1,\dots,t_p)$ defined by $t_i=s_i+i-1$ for $1\le i\le p$. We claim that this defines a bijection between non-decreasing sequences of length $p$ in $\{1,\dots N\}$ and increasing sequences of length $p$ in $\{1,\dots,N+p\}$. Its inverse is indeed defined by mapping $(t_1,\dots,t_p)$ to $(t_1,t_2-1,\dots,t_p-p-1)$. Now increasing sequences of length $p$ in $\{1,\dots,N+p\}$ are subsets of size $p$ of this set. Thus there are $\binom{N+p}{p}$ such sequences. 

A well known bound on the binomial coefficient $\binom nk$ is $(en/k)^k$. Thus $\binom{N+p}{p}\le (e(1+N/p))^p$. 
\qed\end{proof}

Proposition~\ref{prop:Sgi} and Lemma~\ref{lemma:S} improve the bound on $h_n$ given in Section~\ref{sec:first}. Consequently,
we obtain a tighter bound on the number of real roots of a $\SPS(k,m,t)$ polynomial.

\begin{theorem}\label{thm:support} 
Let $\phi\in
\SPS(k,m,t)$. Then $\phi$ has at most 
\[C\times \left[e\times\left(1+\frac{t^m}{2^{k-1}-1}\right)\right]^{2^{k-1}-1}\] real roots, 
where $C$ is a universal constant.
\end{theorem}

\begin{proof}
As in Section~\ref{sec:first}, we combine Proposition~\ref{prop:Boundrphi} with the bound we have just obtained for $h_n$. 
Recall that \[r(\phi)\le 2h_k+4\sum_{i=1}^{k-1} h_i + 2m(2k-1)(t-1)-k.\]
Moreover
the polynomials $f_j$ in a $\SPS(k,m,t)$ polynomial are $t$-sparse, thus $|S|=\left|(\sum_j S(f_j)) -\textbf{1}\right|\leq t^m$. 
We can combine Proposition~\ref{prop:Sgi} and Lemma~\ref{lemma:S} with $S$ and $p=2^{k-1}-1$ to obtain
$h_k\le \left[e\times\left(1+\frac{t^m}{2^{k-1}-1}\right)\right]^{2^{k-1}-1}$.
Since it dominates the other terms of the sum when $t$ grows, this 
proves the theorem.
\qed\end{proof} 

The bound of Lemma~\ref{lemma:S} is reached for a set $S$ of ``far from each other'' integers. More precisely, if the integers in $S$ form a increasing sequence $(s_n)$, such that for all $n$, $ps_n< s_{n+1}$, then $|p\times S|=\binom{p+S}{p}$. Indeed, two different sums of $p$ integers of $S$ cannot have the same value in this case. If this condition is not satisfied,
one can build a set $S$, whose two different sums of $p$ terms have the same value.

In the proof of Theorem~\ref{thm:support}, $S$ is built from the supports of the $f_j$'s. In this case, the preceding discussion shows that if the degrees of the $f_j$'s are not very far from each other, we can improve our bound. In particular, it can be shown that if the monomials of the $f_j$'s are clustered, and each cluster has a constant diameter, then $t^m$ can be replaced by the number of cluster in the statement of the theorem.

\section{Lower bounds}\label{sec:lower}

In this section we introduce a subclass $\mSPS(k,m)$ of the class of 
``easy to compute'' multivariate polynomial families, 
and we use the results of Section \ref{sec:first} 
to show that it does not contain the permanent family.
The polynomials in a $\mSPS(k,m)$ family 
have the same structure as the univariate polynomials in 
the class $\SPS(k,m,t)$ from Definition~\ref{SPS}.
In this section, polynomial families are denoted by their general term in brackets: 
The polynomial $P_n$ is the $n$-th polynomial of the family $(P_n)$. 
When there is no ambiguity on the
number of variables, we denote by $\vec{X}$ the tuple of variables of a polynomial $P_n$.

\begin{definition} \label{mSPS}
We say that a sequence of polynomials $(P_n)$ 
is in $\mSPS(k,m)$ if there is a polynomial
$Q$ such that 
for all $n$:
\begin{itemize}
 \item[(i)]  
$P_n$ depends on at most $Q(n)$ variables.
 \item[(ii)] $P_n(\vec{X}) = \sum_{i=1}^k \prod_{j=1}^{m} f_{jn}^{\alpha_{ij}}(\vec{X})$
 \item[(iii)] The bitsize of $\alpha_{ij}$ is bounded by $Q(n)$.
 \item[(iv)] For all $1\le j\le m$, the 
    polynomial $f_{jn}$ has a constant free circuit of size $Q(n)$ and is $Q(n)$-sparse.
\end{itemize}
\end{definition}

\begin{remark}
 If $(P_n) \in \mSPS(k,m)$ then each $P_n$ has a constant free circuit 
of size polynomial in $n$.
Indeed from the constant free circuits of the polynomials $f_{jn}$ we can build a constant free circuit for $P_n$.
We have to take the $\alpha_{ij}$-th power of $f_{jn}$, which can be done with a circuit of size polynomial 
in the bitsize of $\alpha_{ij}$ thanks to fast exponentiation. The size of the final circuit is up to a constant the sum 
of the sizes of these powering circuits and of the circuits giving  $f_{jn}$, which is thus polynomial in $n$. 
\end{remark}

\begin{definition}
The Pochhammer-Wilkinson polynomial of order $2^n$ is defined
by $\PW_n= \displaystyle{\prod_{i=1}^{2^n} (X-i)}$.
\end{definition}

\begin{definition}
The Permanent over $n^2$ variables is defined by $\PER_n = \displaystyle{\sum_{\sigma \in \Sigma_n} \prod_{i=1}^{n} X_{i\sigma(i)}}$ 
where $\Sigma_n$ is the set of permutations of $\{1,\dotsc,n\}$.
\end{definition}

We now give a lower bound on the Permanent, using its completeness for $\VNP$ \cite{valiant1979completeness},
a result of B\"urgisser on the Pochhammer-Wilkinson polynomials \cite{burgisser2009defining} and our bound on the roots
of the polynomials in $\SPS(k,m,t)$.

\begin{theorem} \label{lb}
 The family of polynomials $(\PER_n)$ is not in $\mSPS(k,m)$ for any $k$ and $m$, 
i.e., there is no representation of the permanent family of the form
$$\PER_n(\vec{X})=
\sum_{i=1}^k \prod_{j=1}^{m} f_{jn}^{\alpha_{ij}}(\vec{X})$$
where the bitsize of the $\alpha_{ij}$, the sparsity of the polynomials 
$f_{jn}$ and their constant-free arithmetic circuit complexity are 
all bounded by a polynomial function $Q(n)$.
\end{theorem}

\begin{proof}
Assume 
by contradiction
that $(\PER_n) \in \mSPS(k,m)$. 
By the previous remark, 
this implies that 
$\PER_n$ can be computed by polynomial
size constant free arithmetic circuits.
As in the proofs of Theorem $4.1$ and $1.2$ in \cite{burgisser2009defining},
it follows from this property 
that there is a family 
$(G_n(X_0,\ldots,X_n))$ 
in 
$\VNP$ 
such that 
\begin{equation} \label{PWtoG}
\PW_n(X)= G_{n}(X^{2^0},X^{2^1},\dots,X^{2^n}).
\end{equation}

Since the permanent is complete for $\VNP$, we have a polynomial $h$ such that 
\begin{equation} \label{GtoPER}
\PER_{h(n)}(z_1,\dots,z_{h(n)^2}) = G_n(X_0,\ldots,X_n) 
\end{equation}
where the $z_i$'s are either 
variables of $G_n$ or constants.
By hypothesis $(\PER_n) \in \mSPS(k,m)$. 
Let $Q$ be the corresponding polynomial from Definition~\ref{mSPS}.
From this definition and from~(\ref{PWtoG}) and~(\ref{GtoPER})
we have
$$\PW_n(X)= \displaystyle{ \sum_{i=1}^k \prod_{j=1}^{m} f_{jn}(X)^{\alpha_{ij}}}$$
where $f_{jn}(X)$ is 
$Q(h(n))$-sparse.
This shows that the polynomial $\PW_n$ is in $\SPS(k,m,R(n))$ 
where $R(n)=Q(h(n))$.

We have proved in Theorem~\ref{thm:naive} that polynomials in $\SPS(k,m,R(n))$
have at most  
$r(n)=C\times((m+2)R(n))^m)^{2^{k-1}-1}$ real roots. 
On the other hand, by construction the polynomial $\PW_n$ has $2^n$ roots, which is larger than $r(n)$ 
for all 
large enough  
$n$. 
This yields a contradiction and completes the proof of the theorem.
\qed\end{proof} 

\begin{remark}
It is possible to relax condition (iv) in Definition~\ref{mSPS}. 
We can replace it by the less restrictive condition:
\emph{\begin{itemize}
\item[(iv')] the polynomial $f_{jn}$ is $Q(n)$-sparse,
\end{itemize}}
i.e., we allow polynomials $f_{jn}$ with arbitrary complex coefficients.
Theorem~\ref{lb} still applies to this larger 
version of the class 
$\mSPS(k,m)$, but for the proof to go through we need to assume the 
Generalized Riemann Hypothesis.
The only change is at the beginning of the proof: Assuming 
that the permanent
family  belongs to the (redefined) class $\mSPS(k,m)$, 
we can conclude that this family 
can be computed by polynomial size arithmetic circuits 
with arbitrary constants. To see this, note
that any non-multilinear monomial in any $f_{jn}$ can be deleted since it
cannot contribute to the final result (the permanent is multilinear).
And since $f_{jn}$ is sparse, there is a polynomial size arithmetic circuit with
arbitrary constants to compute its multilinear monomials. 
The remainder of the proof is essentially unchanged. 
But to deal with arithmetic circuits with arbitrary constants (from the complex field) 
instead of constant-free arithmetic circuits, we shall use 
Corollary $4.2$ of~\cite{burgisser2009defining} instead of Theorems~1.2 and 4.1.
This means that we have to assume GRH as in this corollary.
It is an intriguing question whether this assumption can be removed from 
Corollary $4.2$ of~\cite{burgisser2009defining} and from this lower bound result.
\end{remark}

\section{Polynomial Identity Testing}\label{sec:PIT}

This section is devoted to 
a proof 
that 
Identity Testing can be done in deterministic polynomial
time on the polynomials studied in the previous sections.
Recall from Definition~\ref{def:phin} that for $\phi=\sum_{i=1}^k P_i\in\SPS(k,m,t)$, $(\phi_n)$ is 
defined by $\phi_n=\sum_{i=n}^k g_i^{(n)}P_i$.
 
\begin{lemma}\label{lemma:nullphi}
Let $\phi\in\SPS(k,m,t)$ and $(\phi_n)$ 
as in Definition~\ref{def:phin}.
Then for $l<k$, 
$\phi_l\equiv 0$ 
if and only if 
$\phi_{l+1}\equiv 0$ 
and $\phi_l$ has a smaller degree than $g_l^{(l)}P_l$. 
\end{lemma}
\begin{proof}
If for all $i$, $g_i^{(l)}$ is identically zero,
then the lemma holds. If there is at least one which is not identically zero,
assume that it is $g_l^{(l)}$ up to a reindexing of the terms.

Let $T_l=g_l^{(l)}P_l$, recall that $\phi_{l+1}= g_lT_l\pi(\phi_l/T_l)'$. If 
$\phi_l \equiv 0$, then $\phi_{l+1} \equiv 0$. Moreover, 
we have assumed that $T_l\not\equiv 0$ and it it thus of larger
degree than $\phi_l$ which is identically $0$.

Assume now that $\phi_{l+1}\equiv 0$, that is $g_lT_l\pi(\phi_l/T_l)' \equiv 0$.
By hypothesis,  $T_l$  and $\pi$ are not 
identically zero,
therefore $(\phi_l/T_l)' \equiv 0$. Thus there is $\lambda \in \RR$ such that 
$\phi_l = \lambda T_l$. Since by hypothesis $\phi_l$ and $T_l$ have different degrees,
 $\lambda = 0$ and $\phi_l \equiv 0$.
\qed\end{proof} 

To solve PIT, we will need to explicitly compute the sequence of polynomials $\phi_l$.
Thus, the algorithm is not 
black-box:
it must have  access to a representation of the input polynomial
under form~\eqref{expression}.

\begin{theorem}\label{thm:pit}
Let $k$ and $m$ be two integers and 
$\phi \in 
\SPS(k,m,t)$:
we have
$\phi=\sum_{i=1}^k\prod_{j=1}^m f_j^{\alpha_{ij}}$ where for all $i$ and $j$, $f_j$ is $t$-sparse and $\alpha_{ij}\ge 0$.
Then one can test if $\phi$
is identically zero in time polynomial 
in $t$, in the size of the sparse representation of the $f_j$'s and in the $\alpha_{ij}$'s.
\end{theorem}

\begin{proof}
 Let $(\phi_n)$ 
 be the sequence defined from $\phi$ as in Definition~\ref{def:phin}. 
 Lemma~\ref{lemma:nullphi} 
 implies that $\phi$ is identically zero 
if and only if $\phi_k$ is identically zero and that for all $l < k$,
$\displaystyle{\phi_l = \sum_{i=l}^k g_i^{(l)} P_i}$ has a strictly smaller degree than $g_l^{(l)}P_l$. 
We also assume that $g_l^{(l)}P_l$ is of highest degree amongst the $g_i^{(l)}P_i$
 (always true up to a reordering of these terms).

One can compute the sparse polynomials $g_i^{(l)}$, for all $i$ and $l$ in time polynomial in the size 
of the $f_j$'s if $k$ and $m$ are fixed. 
For each $l$, one can test if the degree of $g_l^{(l)}P_l$ and of $\phi_l$ differ.
One only has to compute the highest degree monomials of each $g_i^{(l)}P_i$ for $i \geq l$.
One can do that in time polynomial in the $\alpha_{ij}$ (not their bitsize) and the size of the $f_j$'s.

Finally, $\phi_k = g_k^{(k)} P_k$ therefore it is identically zero if and only if $g_k^{(k)}$
is identically zero and we have computed it explicitly.
\qed\end{proof}

This algorithm is polynomial in the $\alpha_{ij}$'s,
though ideally we would like it to be polynomial in their bitsize.
\begin{proposition}
Assume that we have access to an oracle which decides whether 
\begin{equation} \label{oracle}
\sum_{i=1}^k \prod_{j=1}^m a_{ij}^{\alpha_{ij}}=0.
\end{equation}
Let $\phi=\sum_{i=1}^k\prod_{j=1}^m f_j^{\alpha_{ij}}$ as in Theorem~\ref{thm:pit}. 
Then one can decide 
deterministically
whether $\phi$ is identically zero in time polynomial in the sparsity of the $f_j$'s and in the bitsize
of the $a_{ij}$'s and $\alpha_{ij}$'s.
\end{proposition}
\begin{proof}
The only dependency in the $\alpha_{ij}$'s in the proof of Theorem~\ref{thm:pit}
is the computation of the coefficient of the highest degree monomials of the $g_i^{(l)}P_i$.
With the oracle for~\eqref{oracle}, we skip this step and achieve a polynomial dependency in the 
bitsize of the $\alpha_{ij}$'s. \qed
 
\end{proof}

A direct computation of the constant on the left-hand side of~\eqref{oracle}
is not possible since it involves
numbers of exponential bitsize (the exponents $\alpha_{ij}$ are given in
binary notation).
The test to 0 can be made by computing modulo random primes, but this is ruled
out since we want a deterministic algorithm.
Note also that 
this test
is a PIT problem for polynomials in 
$\SPS(k,m,t)$ where the $f_j$'s are constant polynomials.
For general arithmetic circuits, it is likewise known that PIT reduces
to the case of circuits
without any 
variable occurrence~(\cite{ABKB09}, Proposition~2.2).

The polynomial identity test from Theorem~\ref{thm:pit}
can also be applied to the class of multivariate 
polynomial families $\mSPS(k,m)$ introduced in the previous section.
Indeed, let $P(X_1,\dots,X_n)=\sum_i\prod_j f_j^{\alpha_{ij}}$ belongs to some 
$\mSPS(k,m)$ family, and suppose we know a bound $d$ on its degree.
We turn $P$ into a univariate polynomial $Q$ by the
classical substitution (sometimes attributed to Kronecker)
$ X_i \mapsto X^{{(d+1)}^i}$.
We write $Q(X)=\sum_i\prod_j g_j^{\alpha_{ij}}$,
where each univariate polynomial $g_j$ is the image of $f_j$ by the substitution.
It is a folklore result that $P\equiv 0$
if and only if $Q\equiv 0$, thus we can apply the PIT algorithm of Theorem~\ref{thm:pit} on $Q$.

Let $s$ be  the size of the representation of $P$, 
meaning that $P$ depends on at most $s$ variables, the $f_j$'s have a 
constant free circuit of size at most $s$ and are $s$-sparse, and the 
$\alpha_{ij}$ are at most equal to $s$. (Note that we do not bound their 
bitsizes but their values as it is needed for our PIT algorithm.)  Then the degree 
of the $f_j$'s is at most $2^s$, and  $d\le 2^{\poly(s)}$ where $\poly(s)$ denotes some
polynomial function of $s$. The $g_j$'s therefore have a degree at most $2^{s\poly(s)}\times 2^s
=2^{s\poly(s)+s}$. This proves that $Q$ satisfies the hypothesis of Theorem~\ref{thm:pit}.

\section{Conclusion}

We have shown that the real $\tau$-conjecture from~\cite{koiran2011shallow} 
holds true for a restricted class of polynomials, and from this result we have
obtained an identity testing algorithm and a lower bound for the permanent.
Other simple cases of the conjecture remain open. 
In the general case, we can expand a sum of product of sparse polynomials
as a sum of at most $kt^m$ monomials. There are therefore at most
$2kt^m-1$ real roots. As pointed out in~\cite{koiran2011shallow},
 the case $k=2$ is already open: is there a polynomial bound on the number
of real roots in this case?
Even simpler versions of this question are open. For instance, we can
ask whether the number of real roots of an expression of the form
$f_1\cdots f_m +1$ is polynomial in $m$ and $t$. A bare bones version
of this problem was pointed out by Arkadev Chattopadhyay 
(personal communication): taking $m=2$, we can
ask what is the maximum number of real roots of an expression of the
form $f_1f_2+1$. Expansion as a sum of monomials yields a $O(t^2)$ 
upper bound, 
but for all we know the true bound could be $O(t)$.


\begin{thebibliography}{10}

\bibitem{AS09}
M.~Agrawal and R.~Saptharishi.
\newblock {Classifying Polynomials and Identity Testing}.
\newblock In {\em Current Trends in Science}, pages 149--162. Indian Academy of
  Sciences, 2009.

\bibitem{agrawal2008arithmetic}
M.~Agrawal and V.~Vinay.
\newblock Arithmetic circuits: A chasm at depth four.
\newblock In {\em Proceedings of the 49th Annual IEEE Symposium on Foundations
  of Computer Science}, pages 67--75, 2008.

\bibitem{ABKB09}
E.~Allender, P.~B\"urgisser, J.~Kjeldgaard-Pedersen, and P.~Bro-Miltersen.
\newblock On the complexity of numerical analysis.
\newblock {\em {SIAM} Journal on Computing}, 38(5):1987--2006, 2009.
\newblock Conference version in CCC 2006.

\bibitem{beecken2011algebraic}
M.~Beecken, J.~Mittmann, and N.~Saxena.
\newblock Algebraic independence and blackbox identity testing.
\newblock {\em Proceedings of the 38th International Colloquium on Automata,
  Languages and Programming}, 2011.
\newblock Arxiv preprint arXiv:1102.2789.

\bibitem{BC76}
A.~Borodin and S.~Cook.
\newblock On the number additions to compute specific polynomials.
\newblock {\em {SIAM} Journal on Computing}, 5(1):146--157, 1976.

\bibitem{burgisser2009defining}
P.~B{\"u}rgisser.
\newblock {On defining integers and proving arithmetic circuit lower bounds}.
\newblock {\em Computational Complexity}, 18(1):81--103, 2009.

\bibitem{BCS}
P.~B\"urgisser, M.~Clausen, and M.~A. Shokrollahi.
\newblock {\em Algebraic Complexity Theory}.
\newblock Springer, 1997.

\bibitem{HS82}
J.~Heintz and C.-P. Schnorr.
\newblock Testing polynomials which are easy to compute.
\newblock In {\em Logic and Algorithmic (an International Symposium held in
  honour of {Ernst Specker})}, pages 237--254. Monographie $n^{\tiny o}$ 30 de
  L'Enseignement Math\'ematique, 1982.
\newblock Preliminary version in {\em Proc. 12th {ACM} Symposium on Theory of
  Computing}, pages~262-272, 1980.

\bibitem{kabanets2004derandomizing}
V.~Kabanets and R.~Impagliazzo.
\newblock Derandomizing polynomial identity tests means proving circuit lower
  bounds.
\newblock {\em Computational Complexity}, 13(1):1--46, 2004.

\bibitem{koiran2010arithmetic}
P.~Koiran.
\newblock Arithmetic circuits: the chasm at depth four gets wider.
\newblock {\em Arxiv preprint arXiv:1006.4700}, 2010.

\bibitem{koiran2011shallow}
P.~Koiran.
\newblock {Shallow circuits with high-powered inputs}.
\newblock {\em Proceedings of the Second Symposium on Innovations in Computer
  Science}, 2011.

\bibitem{li2003counting}
T.Y. Li, J.M. Rojas, and X.~Wang.
\newblock Counting real connected components of trinomial curve intersections
  and m-nomial hypersurfaces.
\newblock {\em Discrete and computational geometry}, 30(3):379--414, 2003.

\bibitem{Sax09}
N.~Saxena.
\newblock {Progress on Polynomial Identity Testing}.
\newblock {\em Bull. EATCS}, 99:49--79, 2009.

\bibitem{Schw}
J.~T. Schwartz.
\newblock Fast probabilistic algorithms for verification of polynomials
  identities.
\newblock {\em Journal of the ACM}, 27:701--717, 1980.

\bibitem{ShubSmale}
M.~Shub and S.~Smale.
\newblock On the intractability of {Hilbert}'s {Nullstellensatz} and an
  algebraic version of ``{P=NP}".
\newblock {\em Duke Mathematical Journal}, 81(1):47--54, 1995.

\bibitem{smale1998mathematical}
S.~Smale.
\newblock Mathematical problems for the next century.
\newblock {\em The Mathematical Intelligencer}, 20(2):7--15, 1998.

\bibitem{valiant1979completeness}
L.G. Valiant.
\newblock {Completeness classes in algebra}.
\newblock In {\em Proceedings of the 11th Annual ACM Symposium on Theory of
  Computing}, pages 249--261, 1979.

\end{thebibliography}
\end{document}